\documentclass[a4paper,12pt]{article}

\usepackage{xcolor}
\usepackage[shortlabels]{enumitem}

\usepackage[utf8]{inputenc}
\usepackage{indentfirst}
\usepackage{amsbsy,amssymb,amsthm,amsmath, amsfonts,fixltx2e,bbding} 
\usepackage{graphicx}
\usepackage{ragged2e}
\makeatletter
\newcommand*{\rom}[1]{\expandafter\@slowromancap\romannumeral #1@}
\makeatother

\usepackage{graphicx,amssymb,amsmath,tikz}
\usetikzlibrary{shapes,snakes,positioning}

\usepackage[ruled, lined, linesnumbered, longend]{algorithm2e}

\title{Infinite separation between general and chromatic memory}

\author{Alexander Kozachinskiy\footnote{The author is funded by ANID - Millennium Science Initiative Program - Code ICN17002, and the National Center for Artificial Intelligence CENIA
FB210017, Basal ANID.} \\ IMFD \& CENIA, Chile\\ alexander.kozachinskyi@cenia.cl}

\newtheorem{definition}{Definition}
\newtheorem{theorem}{Theorem}
\newtheorem{lemma}{Lemma}

\newtheorem{proposition}{Proposition}

\newtheorem{remark}{Remark}
\newtheorem{question}{Question}
\newtheorem{problem}{Problem}
\newtheorem{fact}{Fact}

\newcommand{\col}{\mathsf{col}}

\newcommand{\genmem}{\mathsf{GenMem}}
\newcommand{\chrmem}{\mathsf{ChrMem}}
\newcommand{\RL}{\mathsf{RL}}
\newcommand{\zero}{\mathbf{\overline{0}}}
\newcommand{\ant}{\mathsf{ant}}
\newcommand{\layer}{\mathsf{layer}}
\newcommand{\fp}{\mathsf{FinitePlays}}
\newcommand{\ip}{\mathsf{InfinitePlays}}
\newcommand{\source}{\mathsf{source}}
\newcommand{\target}{\mathsf{target}}

\begin{document}
\maketitle

\begin{abstract}
In this paper, we construct a winning condition $W$ over a finite set of colors such that, first, every finite arena has a strategy with 2 states of general memory which is optimal w.r.t.~$W$, and second, there exists no $k$ such that every finite arena has a strategy with $k$ states of chromatic memory which is optimal w.r.t.~$W$.
\end{abstract}
\section{Introduction}

Memory requirements for games on graphs have been studied for decades. Initially, these studies were motivated by applications to automata theory and the decidability of logical theories. For example, the memoryless determinacy of parity games is a key ingredient for the complementation of tree automata and leads to the decidability of
the monadic second-order theory of trees~\cite{zielonka1998infinite}.  Recently, games on graphs have become an important tool in reactive synthesis~\cite{bloem2018graph}. They serve there as a model of the interaction between a reactive system and the environment. One question studied in games on graphs is which winning conditions admit ``simple'' winning strategies. The prevailing measure of the complexity of strategies in the literature is memory. In this note, we study two kinds of memory -- \emph{general} (a.k.a.~\emph{chaotic}) memory and \emph{chromatic} memory. The relationship between them was first addressed in the Ph.D.~thesis of Kopczy\'{n}ski~\cite{phdthesis}, followed by several recent works~\cite{bouyer_et_al:LIPIcs:2020:12836,casares:LIPIcs.CSL.2022.12,casares2022size}.

\medskip

We focus on deterministic games, infinite-duration and turn-based. We call our players Protagonist and Antagonist. They play over a finite\footnote{There are papers that study these games over infinite graphs, but in this note we only work with finite graphs.} directed graph called an \emph{arena}. Its set of nodes has to be partitioned into ones controlled by Protagonist and ones controlled by Antagonist. Players move a token over the nodes of the graph along its edges. In each turn, the token is moved by the player controlling the current node.

After infinitely many turns, this process produces an infinite path in our graph. A \emph{winning condition} is a set of infinite paths that are winning for Protagonist. In the literature, a standard way of defining winning conditions assumes that arenas are edge-colored by elements of some set of colors $C$. Then any subset $W\subseteq C^\omega$ is associated with a winning condition, consisting of all infinite paths whose sequence of colors belongs to $W$. 

In this paper, we seek simple winning strategies of Protagonist, while the complexity of Antagonist's strategies is mostly irrelevant to us. Such asymmetry is motivated by reactive synthesis, where Protagonist represents a system and Antagonist represents the environment. Now, 
the main measure of the complexity of Protagonist's strategies for us is memory. Qualitatively,  we distinguish between \emph{finite-memory} strategies and \emph{infinite-memory} strategies. In turn, among finite-memory strategies, we prefer those that have fewer states of memory.

Finite-memory strategies are defined through so-called \emph{memory structures}. Intuitively, a memory structure plays the role of a ``hard disk'' of a strategy. Formally,  a \emph{general} memory structure $\mathcal{M}$ is a deterministic finite automaton whose input alphabet is the set of edges of an arena. During the game, edges over which the token moves are fed to $\mathcal{M}$ one by one. Correspondingly, the state of $\mathcal{M}$ is updated after each move.
 Now, a strategy \emph{built on top of a memory structure} $\mathcal{M}$ (or simply an $\mathcal{M}$-strategy) is a strategy whose moves at any moment depend solely on two things: first, the current arena node, and second, the current state of $\mathcal{M}$. A strategy is finite-memory if it can be built on top of some memory structure. More precisely, if this memory structure has $k$ states, then strategies built on top of it are strategies \emph{with $k$ states of general memory}. Of course, some strategies cannot be built on top of any memory structure. Such strategies are infinite-memory strategies.

We also consider a special class of general memory structures called \emph{chromatic} memory structures. A memory structure is chromatic if its transition function does not distinguish edges of the same color. In other words, chromatic memory structures only reads colors of edges that are fed into them.
Alternatively, a chromatic memory structure can be viewed as a finite automaton whose input alphabet is not the set of edges, but the set of colors.
Correspondingly, strategies that are built on top of a chromatic memory structure with $k$ states are called strategies with \emph{$k$ states of chromatic memory}.

\subsection*{Around a Kopczy\'{n}ski's question}
 Complexity of strategies brings us to complexity of winning conditions. For a given winning condition, we want to determine the minimal amount of memory which is sufficient to win whenever it is possible to win. More specifically, the \textbf{general memory complexity} of a winning condition $W$, denoted by $\genmem(W)$,
 is the minimal $k\in\mathbb{N}$ such that in every arena there exists a Protagonist's strategy $S$ with $k$ states of general memory which is optimal w.r.t.~$W$. If no such $k$ exists, we set $\genmem(W) = +\infty$. Now, ``$S$ is optimal w.r.t.~$W$'' means that there exists no node $v$ such that some Protagonist's strategy is winning from $v$ w.r.t.~$W$ and $S$ is not.  Substituting ``general memory'' by ``chromatic memory'', we obtain a definition of the \textbf{chromatic memory complexity} of $W$, which is denoted by $\chrmem(W)$.

For any $W$, we have $\genmem(W)\le \chrmem(W)$. Our paper revolves around a question from the Ph.D.~thesis of Kopczy\'{n}ski~\cite{phdthesis}.
\begin{question}
\label{kop_conj} Is this true that $\genmem(W) = \chrmem(W)$ for every winning condition $W$?
\end{question}

To understand Kopczy\'{n}ski's motivation, we first have to go back to 1969, when B\"{u}chi and Landweber~\cite{buchi1969solving} established  that $\chrmem(W)$ is finite for all $\omega$-regular $W$. An obvious corollary of this is that $\genmem(W)$ is also finite for all $\omega$-regular $W$. Since then, there is an unfinished quest of \emph{exactly characterizing} $\chrmem(W)$ and $\genmem(W)$ for $\omega$-regular $W$. In particular, it is open whether $\chrmem(W)$ and $\genmem(W)$ are computable given an $\omega$-regular $W$ as an input (assuming $W$ is given, say, in a form of a non-deterministic  B{\"u}chi automaton recognizing $W$).

In his Ph.D.~thesis, Kopczy\'{n}ski contributed to this question by giving an algorithm computing $\chrmem(W)$ for prefix-independent $\omega$-regular $W$  (a winning condition is called prefix-independent if it is invariant under adding and removing finite prefixes). Prior to that, he  published a weaker version of this result in~\cite{kopczynski2007omega}.
 He asked Question \ref{kop_conj} to find out, whether his algorithm also computes $\genmem(W)$ for prefix-independent $\omega$-regular $W$. His other motivation was that the same chromatic memory structure can be used in different arenas. Indeed, transition functions of chromatic memory structures can be defined over colors so that we do not have to specify them individually for each arena.

Question \ref{kop_conj} was recently answered by Casares in~\cite{casares:LIPIcs.CSL.2022.12}. Namely, for every $n\in\mathbb{N}$ he gave a \emph{Muller} condition $W$ over $n$ colors with $\genmem(W) = 2$ and $\chrmem(W) = n$.

\begin{definition}
A winning condition $W\subseteq C^\omega$ is \textbf{Muller} if $C$ is finite and $\alpha\in W \iff\beta\in W$ for any two $\alpha, \beta\in C^\omega$ that have the same sets of colors occurring infinitely often in them.
\end{definition}
Every Muller condition is prefix-independent and $\omega$-regular. Hence, we now know that Kopczy\'{n}ski's algorithm does not always  compute $\genmem(W)$ for prefix-independent $\omega$-regular $W$. It is still open whether some other algorithm does this job.

In a follow-up work, Casares, Colcombet and Lehtinen~\cite{casares2022size} achieve a larger gap between $\genmem(W)$ and $\chrmem(W)$. Namely, they construct a Muller $W$ over $n$ colors such that $\genmem(W)$ is linear in $n$ and $\chrmem(W)$ is exponential in $n$.

It is worth mentioning that Casares, Colcombet and Lehtinen derive these examples from their new automata-theoretic characterizations of $\chrmem(W)$ and $\genmem(W)$ for Muller $W$. First, Casares~\cite{casares:LIPIcs.CSL.2022.12} showed that $\chrmem(W)$ equals the minimal size of a deterministic Rabin automaton, recognizing $W$, for every Muller $W$. Second, Casares, Colcombet and Lehtinen~\cite{casares2022size} showed that $\genmem(W)$ equals the minimal size of a good-for-games Rabin automaton, recognizing $W$, for every Muller $W$. The latter result complements an earlier work by Dziembowski, Jurdzi\'{n}ski and Walukiewicz~\cite{dziembowski1997much}, who characterized $\genmem(W)$ for Muller $W$ in terms of their Zielonka's trees~\cite{zielonka1998infinite}.

These examples, however, do not answer a natural follow-up question -- can the gap between $\genmem(W)$ and $\chrmem(W)$ be infinite? To answer it, we have to go beyond Muller and even $\omega$-regular conditions (because $\chrmem(W)$ is finite for them). 

\begin{question}
\label{my_conj}
Is it true that for every \textbf{finite} set of colors $C$ and for every winning condition $W\subseteq C^\omega$ we have $\genmem(W) < +\infty \implies \chrmem(W) < +\infty$?
\end{question}
\begin{remark}
 If we do not insist on finiteness of $C$, a negative answer to Question \ref{my_conj} follows from the example of Casares. Namely, for every $n$ he defines a winning condition $W_n\subseteq\{1, 2,\ldots n\}^\omega$, consisting of all $\alpha\in\{1, 2,\ldots n\}^\omega$ such that there are exactly two numbers from 1 to $n$ that occur infinitely often in $\alpha$. He then shows that $\genmem(W_n) = 2$ and $\chrmem(W_n) = n$ for every $n$. We can now consider the union of these winning conditions $\cup_{n\ge 2} W_n$, which is a winning condition over $C = \mathbb{N}$. On one hand, $\genmem(\cup_{n\ge 2} W_n) = 2$ because every arena has only finitely many natural numbers as colors, and hence $\cup_{n\ge 2} W_n$ coincides with $W_n$ for some $n$ there. On the other hand, we have $\chrmem(\cup_{n\ge 2} W_n) \ge\chrmem(W_n) = n$ for every $n$, which means that $\chrmem(\cup_{n\ge 2} W_n) = +\infty$.
\end{remark}

In this paper, we answer negatively to Question \ref{my_conj}.
\begin{theorem}
\label{uniform_separation}
There exists a finite set of colors $C$ and a winning condition $W\subseteq C^\omega$ such that $\genmem(W) = 2$ and $\chrmem(W) = +\infty$.
\end{theorem}

Topologically, our $W$ belongs to the $\Sigma_2^0$-level of the Borel hierarchy.
Next, the size of $C$ in our example is 5, and there is a chance that it can be reduced. In turn, $\genmem(W)$ is optimal because $\genmem(W) = 1$ implies $\chrmem(W) = 1$ (one state of general memory is equally useless as one state of chromatic memory).

 We call our $W$ the ``Rope Ladder'' condition. We define it in Section \ref{sec:def}. The upper bound on $\genmem(W)$ and the lower bound on $\chrmem(W)$ are given in Section \ref{sec:upp} and in Section \ref{sec:low}, respectively. Before that, we give Preliminaries in Section \ref{sec:prel}.

\subsection*{Further open questions}

Still, some intriguing variations of Question \ref{my_conj} remain open. For example, it is interesting to obtain Theorem \ref{uniform_separation} for a closed condition, i.e.~a condition in the $\Pi^0_1$-level of the Borel hierarchy, or equivalently, a condition given by a set of prohibited finite prefixes. In the game-theoretic literature, such conditions are usually called safety conditions. Our $W$ is an infinite union of safety conditions. In~\cite{colcombet2014playing}, Colcombet, Fijalkow and Horn give a characterization $\genmem(W)$ for safety $W$. Recently, Bouyer,  Fijalkow, Randour,and Vandenhove~\cite{regular} obtained a characterization of $\chrmem(W)$ for safety $W$.


\begin{problem}
\label{pref_ind}
Construct a finite set of colors $C$ and a safety winning condition $W\subseteq C^\omega$ such that $\genmem(W) < \infty$ and $\chrmem(W) = +\infty$.
\end{problem}

It is equally interesting to obtain Theorem \ref{uniform_separation}
 for a prefix-independent $W$, as our $W$  is not prefix-independent. 

\begin{problem}
\label{pref_ind}
Construct a finite set of colors $C$ and a prefix-independent winning condition $W\subseteq C^\omega$ such that $\genmem(W) < \infty$ and $\chrmem(W) = +\infty$.
\end{problem}

 There is also a variation of Question \ref{my_conj} related to a paper of Bouyer, Le Roux, Oualhadj, Randour, and Vandenhove~\cite{bouyer_et_al:LIPIcs:2020:12836}. In this paper, they introduce and study the class of arena-independent finite-memory determined winning conditions. When the set of colors $C$ is finite, this class can be defined as the class of $W$ such that both $\chrmem(W)$ and $\chrmem(C^\omega\setminus W)$ are finite\footnote{In their original definition, the ``memory structure'' (see Preliminaries) must be  the same in all arenas. When $C$ is finite, this definition is equivalent, because there are just finitely many chromatic memory structures up to a certain size.  If none of them works for all arenas, one can construct a finite arena where none of them works.} (meaning that both Protagonist and Antagonist can play optimally w.r.t.~$W$ using some constant number of states of chromatic memory).

First, Bouyer et al.~obtain an automata-theoretic characterization of arena-independent finite-memory determinacy. Second, they deduce a one-to-two-player lifting theorem from it. Namely, they show that as long as both $\chrmem(W)$ and $\chrmem(\lnot W)$ are finite in arenas without the Antagonist's nodes, the same is true for all arenas.

A natural step forward would be to study conditions $W$ for which 
 both $\genmem(W)$ and $\genmem(\lnot W)$ are finite. Unfortunately, it is even unknown whether this is a larger class of conditions. This raises the following problem.
\begin{problem}
\label{bouyer}
Construct a finite set of colors $C$ and a winning condition $W\subseteq C^\omega$ such that $\genmem(W)$ and $\genmem(\lnot W)$ are finite, but $\chrmem(W)$ is infinite.
\end{problem}
In fact, it is not clear if our $W$ from Theorem \ref{uniform_separation} solves this problem. We do not know whether $\genmem(\lnot W)$ is finite for this $W$.

 Question \ref{my_conj} is also open over infinite arenas. There is a  relevant result  due to Bouyer, Randour and Vandenhove~\cite{bouyer_et_al:LIPIcs.STACS.2022.16}, who showed that the class of $W$ for which $\chrmem(W)$ and $\chrmem(\lnot W)$ are both finite in infinite arenas coincides with the class of $\omega$-regular $W$. Thus, it would be sufficient to give a non-$\omega$-regular $W$ for which both $\genmem(W)$ and $\genmem(\lnot W)$ are finite in infinite arenas.

Finally, let us mention a line of work which studied the relationship between chromatic and general memory in the non-uniform setting. Namely, fix a single arena $\mathcal{A}$ and some winning condition $W$, and then consider two quantities: first, the minimal $k_{gen}$ such that $\mathcal{A}$ has an optimal strategy with $k_{gen}$ states of general memory, and second, the minimal $k_{chr}$ such that  $\mathcal{A}$ has an optimal strategy with $k_{chr}$ states of chromatic memory. In~\cite{le2020time}, Le Roux showed that if $k_{gen}$ is finite, then $k_{chr}$ is also finite. There is no contradiction with Theorem \ref{uniform_separation} because $k_{chr}$ depends not only on $k_{gen}$, but also on $\mathcal{A}$.
 A tight bound on $k_{chr}$ in terms of $k_{gen}$ and the number of nodes of $\mathcal{A}$ was obtained in~\cite{kozachinskiy2022state}.

\section{Preliminaries}
\label{sec:prel}
\textbf{Notation.} For a set $A$, we let $A^*$ and $A^\omega$ stand for the set of all finite and the set of all infinite sequences of elements of $A$, respectively. For $x\in A^*$, we let $|x|$ denote the length of $x$. We also set $|x| = +\infty$ for $x\in A^\omega$. We let $\circ$ denote the function composition. The set of positive integral numbers is denoted by $\mathbb{Z}^+$.
\subsection{Arenas}

\begin{definition} Let $C$ be a non-empty set. A tuple $\mathcal{A} = \langle V_P, V_A, E\rangle$ is called an \textbf{arena over the set of colors $C$} if the following conditions hold:
\begin{itemize}
\item $V_P, V_A, E$ are finite sets such that $V_P\cap V_A = \varnothing$, $V_P \cup V_A\neq\varnothing$ and $E\subseteq (V_P\cup V_A) \times C\times (V_P\cup V_A)$;
\item for every $s\in V_P\cup V_A$ there exist $c\in C$ and $t\in V_P\cup V_A$ such that $(s, c, t)\in E$.
\end{itemize}
\end{definition}

Elements of the set $V = V_P\cup V_A$ will be called nodes of $\mathcal{A}$. Elements of $V_P$ will be called nodes controlled by Protagonist (or simply Protagonist's nodes). Similarly, elements of $V_A$ will be called nodes controlled by Antagonist (or simply Antagonist's nodes). Elements of $E$ will be called edges of $\mathcal{A}$. For an edge $e = (s, c, t) \in E$ we define $\source(e) = s, \col(e) = c$ and $\target(e) = t$.
We imagine $e\in E$ as an arrow which is drawn from the node $\source(e)$ to the node $\target(e)$ and which is colored into $\col(e)$. Note that the second condition in the definition of an arena means that every node has at least one out-going edge.

 We extend the domain of $\col$ to the set $E^*\cup E^\omega$ by
\[\col(e_1 e_2 e_3\ldots) = \col(e_1)\col(e_2)\col(e_3)\ldots, \qquad e_1, e_2, e_3,\ldots\in E.\]

A non-empty sequence $p = e_1 e_2 e_3 \ldots \in E^*\cup E^\omega$ is called a path if for any $1\le i < |p|$ we have $\target(e_i) = \source(e_{i+1})$. We set $\source(p) = \source(e_1)$ and, if $p$ is finite, $\target(p) = \target(e_{|p|})$.
For technical convenience, every node $v\in V$ is assigned a $0$-length path $\lambda_v$, for which we set $\source(\lambda_v) = \target(\lambda_v) = v$ and $\col(\lambda_v) = \mbox{empty string}$. 

Paths are sequences of edges, so we will say that some paths are prefixes of the others. However, we have to define this for $0$-length paths. Namely, we say that $\lambda_v$ is a prefix of a path $p$ if and only if $\source(p) = v$.

\subsection{Strategies}

Let $\mathcal{A} = \langle V_P, V_A, E\rangle$ be an arena over the set of colors $C$. A Protagonist's strategy in $\mathcal{A}$ is any function 
\[S\colon\{p\mid p \mbox{ is a finite path in $\mathcal{A}$ with }\target(p) \in V_P\}\to E,\]
such that for every $p$ from the domain of $S$ we have $\source(S(p)) = \target(p)$. In this paper, we do not mention  Antagonist's strategies, but, of course, they can be defined similarly.

The set of finite paths in $\mathcal{A}$ is the set of positions of the game. Possible starting positions are $0$-length paths $\lambda_s, s\in V$. When the starting position\footnote{We do not have to redefine $S$ for every starting position. The same $S$ can be played from any of them.} is $\lambda_s$, we say that the game starts at $s$. Now, consider any finite path $p$. Protagonist is the one to move after $p$ if and only if $t =\target(p)$ is a Protagonist's node. In this situation, Protagonist must choose some edge starting at  $t$. A Protagonist's strategy fixes this choice for every $p$ with $\target(p)\in V_P$. We then append this edge to $p$ and get the next position in the game. Antagonist acts the same for those $p$ such that $\target(p)$ is an Antagonist's node.

Let us define paths that are consistent with a Protagonist's strategy $S$. First, any 0-length path $\lambda_v$ is consistent with $S$. Now, a non-empty  path $p = e_1 e_2 e_3\ldots$ (which may be finite or infinite) is consistent with  $S$ if the following holds:
\begin{itemize}
\item if $\source(p) \in V_P$, then $e_1 = S(\lambda_{\source(p)})$;
\item for every $1 \le i < |p|$, if $\target(e_i) \in V_P$, then $e_{i + 1} = S(e_1 e_2\ldots e_i)$.
\end{itemize}
For brevity, paths that are consistent with $S$ will also be called \emph{plays with} $S$. For a node $v$, we let $\fp(S, v)$ and $\ip(S, v)$ denote the set of finite plays with $S$ that start at $v$ and the set of infinite plays with $S$ that start at $v$, respectively. For $U\subseteq V$, we define $\fp(S, U) = \bigcup_{v\in U} \fp(S, v)$ and $\ip(S, U) = \bigcup_{v\in U} \ip(S, v)$.

\subsection{Memory structures}

Let $\mathcal{A} = \langle V_P, V_A,  E\rangle$ be an arena over the set of colors $C$.
A memory structure in $\mathcal{A}$ is a tuple $\mathcal{M} = \langle M, m_{init}, \delta\rangle$, where $M$ is a finite set, $m_{init}\in M$ and $\delta\colon M\times E\to M$. Elements of $M$ are called states of $\mathcal{M}$, $m_{init}$ is called the initial state of $\mathcal{M}$ and $\delta$ is called the transition function of $\mathcal{M}$. Given $m\in M$, we inductively define the function $\delta(m,\cdot)$ over arbitrary finite sequences of edges:
\begin{align*}
\delta(m, \mbox{empty sequence}) &= m,\\ 
\delta(m, se) &= \delta(\delta(m, s), e), \qquad s\in E^*, e\in E.
\end{align*}
In other words, $\delta(m, s)$ is the state of $\mathcal{M}$ it transits to from the state $m$ if  we fed $s$ to it.

A memory structure $\mathcal{M} = \langle M, m_{init}, \delta\rangle$ is called chromatic if $\delta(m, e_1) = \delta(m, e_2)$ for every $m\in M$ and for every $e_1, e_2\in E$ with $\col(e_1) = \col(e_2)$. In this case, there exists $\sigma\colon M\times C\to M$ such that $\delta(m, e) = \sigma(m,\col(e))$. In other words, we can view $\mathcal{M}$ as a deterministic finite automaton over $C$, with $\sigma$ being its transition function.

A strategy $S$ is built on top of a memory structure $\mathcal{M}$ if we have $S(p_1) = S(p_2)$ for any two paths $p_1, p_2$ with $\target(p_1) = \target(p_2)$ and $\delta(m_{init}, p_1) = \delta(m_{init}, p_2)$. In this case, we sometimes simply say that $S$ is an $\mathcal{M}$-strategy. To define an $\mathcal{M}$-strategy $S$, it is sufficient to give its \emph{next-move function} $n_S\colon V_P\times M\to E$. For $v\in V_P$ and $m\in M$, the value of $n_S(v, m)$ determines what $S$ does for paths that end at $v$ and bring $\mathcal{M}$ to $m$ from $m_{init}$.

A strategy $S$ built on top of a memory structure $\mathcal{M}$ with $k$ states is called a strategy with $k$ states of general memory. If $\mathcal{M}$ is chromatic, then $S$ is a strategy with $k$ states of chromatic memory.

For brevity, if $S$ is an $\mathcal{M}$-strategy and $p$ is a finite path, we say that $\delta(m_{init}, p)$ is the state of $S$ after $p$.

\subsection{Winning conditions and their memory complexity}
A winning condition is any set $W\subseteq C^\omega$. We say that a Protagonist's strategy $S$ is winning from a node $u$ w.r.t.~$W$ if the image of $\ip(S, u)$ under $\col$ is a subset of $W$.
In other words, any infinite play from $u$ against $S$ must give a sequence of colors belonging to $W$. Now, a Protagonist's strategy $S$ is called optimal w.r.t.~$W$ if there exists no node $u$ such that some Protagonist's strategy is winning from $u$ w.r.t.~$W$ and $S$ is not.

We let $\genmem(W)$ be the minimal $k\in\mathbb{Z}^+$ such that in every arena $\mathcal{A}$ over $C$ there exists a Protagonist's strategy with $k$ states of general memory which is optimal w.r.t.~$W$. If no such $k$ exists, we set $\genmem(W) = +\infty$.
Likewise, we let  $\chrmem(W)$ be the minimal $k\in\mathbb{Z}^+$ such that in every arena $\mathcal{A}$ over $C$ there exists a Protagonist's strategy with $k$ states of general memory which is optimal w.r.t.~$W$. Again, if no such $k$ exists, we set $\chrmem(W) = +\infty$.

\section{The ``Rope Ladder'' Condition}
\label{sec:def}
Consider the partially ordered set $\Omega = (\mathbb{N}\times\{0, 1\}, \preceq)$, where $\preceq$ is defined by
\[
\label{order}
\forall (n, a), (m, b)\in\mathbb{N}\times\{0, 1\}\qquad (n, a)\preceq (m, b) \iff (n, a) = (m, b) \mbox{ or } n < m.
\]

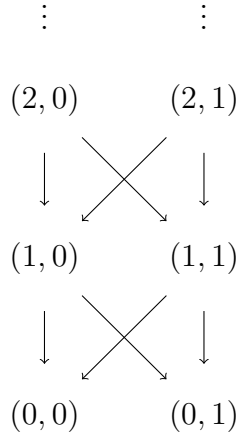
\begin{center}
\begin{tikzpicture}
\node[draw=none,  circle] (00) {$(0, 0)$};
\node[draw=none,  circle, right = 0.7cm of 00] (01) {$(0, 1)$};
\node[draw=none,  circle, above=0.7cm of 00] (10) {$(1, 0)$};
\node[draw=none,  circle, right = 0.7cm of 10] (11) {$(1, 1)$};
\node[draw=none,  circle, above=0.7cm of 10] (20) {$(2, 0)$};
\node[draw=none,  circle, right = 0.7cm of 20] (21) {$(2, 1)$};
\node[draw=none, above = 0.1cm of 20] (dots) {$\vdots$};
\node[draw=none, above = 0.1cm of 21] (dots) {$\vdots$};

\draw[->] (10) -- (00);
\draw[->] (10) -- (01);
\draw[->] (11) -- (00);
\draw[->] (11) -- (01);
\draw[->] (20) -- (10);
\draw[->] (20) -- (11);
\draw[->] (21) -- (10);
\draw[->] (21) -- (11);
\end{tikzpicture}
\end{center}
Above is its Hasse diagram, with arrows representing $\preceq$ (they are directed from bigger elements to smaller elements):

We will use an abbreviation $\zero = (0, 0)$. Next, 
we let $\mathbb{M}$ be the set of all functions $f\colon \Omega\to\Omega$ that are monotone w.r.t.~$\preceq$. Being monotone w.r.t.~$\preceq$ means that $x\preceq y\implies f(x)\preceq f(y)$ for all $x, y\in\Omega$.

\begin{definition}
\label{rope_ladder}
\textbf{The Rope Ladder} condition is a set $\RL\subseteq \mathbb{M}^\omega$, consisting of all infinite sequences $(f_1, f_2, f_3, \ldots)\in \mathbb{M}^\omega$ for which there exists $(N, b)\in\Omega$ such that  $f_{n}\circ\ldots\circ f_2\circ f_1(\zero)\preceq (N, b)$ for all $n\ge 1$.
\end{definition}
We will use the following informal terminology with regard to $\RL$. Imagine that there is an ant which can move over the elements of $\Omega$. Initially, it sits at $\zero$. Next, take any sequence $(f_1, f_2, f_3, \ldots)\in \mathbb{M}^\omega$. We start moving the ant by applying functions from the sequence to the position of the ant. Namely, we first move the ant from $\zero$ to $f_1(\zero)$, then from $f_1(\zero)$ to $f_2\circ f_1(\zero)$, and so on. Now, $(f_1, f_2, f_3, \ldots)\in \RL$ if and only if there exists a ``layer'' in $\Omega$ which is never exceeded by the ant.

\begin{remark}
$\RL$ is defined over infinitely many colors, but for our lower bound on its chromatic memory complexity we will consider its restriction to some finite subset of $\mathbb{M}$.
\end{remark}

To illustrate these definitions, we establish the following fact. It can also be considered as a warm-up for our lower bound.
\begin{fact}
\label{not_positional}
$\chrmem(\RL) > 1$.
\end{fact}
\begin{proof}
First, consider $u, v\colon\Omega\to\Omega$, depicted below:
\begin{center}
\begin{tikzpicture}
\node[draw=none,  circle] (f00) {$(0, 0)$};
\node[draw=none,  circle, right = 0.7cm of f00] (f01) {$(0, 1)$};
\node[draw=none,  circle, above=0.7cm of f00] (f10) {$(1, 0)$};
\node[draw=none,  circle, right = 0.7cm of f10] (f11) {$(1, 1)$};
\node[draw=none,  circle, above=0.7cm of f10] (f20) {$(2, 0)$};
\node[draw=none,  circle, right = 0.7cm of f20] (f21) {$(2, 1)$};
\node[draw=none, above = 0.05cm of f20] (fdots1) {$\vdots$};
\node[draw=none, above = 0.05cm of f21] (fdots2) {$\vdots$};
\node[draw=none, right = 0.6cm of fdots1] (f) {$u$};

\draw[->,red,thick] (f00) -- (f10);
\draw[->,red,thick] (f01) -- (f11);
\draw[->,red,thick] (f10) -- (f20);
\draw[->,red,thick] (f11) -- (f21);

\node[draw=none,  circle, right=5cm of f00] (g00) {$(0, 0)$};
\node[draw=none,  circle, right = 0.7cm of g00] (g01) {$(0, 1)$};
\node[draw=none,  circle, above=0.7cm of g00] (g10) {$(1, 0)$};
\node[draw=none,  circle, right = 0.7cm of g10] (g11) {$(1, 1)$};
\node[draw=none,  circle, above=0.7cm of g10] (g20) {$(2, 0)$};
\node[draw=none,  circle, right = 0.7cm of g20] (g21) {$(2, 1)$};
\node[draw=none, above = 0.05cm of g20] (gdots1) {$\vdots$};
\node[draw=none, above = 0.05cm of g21] (gdots2) {$\vdots$};
\node[draw=none, right = 0.6cm of gdots1] (g) {$v$};

\draw[->,red,thick] (g00) -- (g11);
\draw[->,red,thick] (g01) -- (g10);
\draw[->,red,thick] (g10) -- (g21);
\draw[->,red,thick] (g11) -- (g20);

\end{tikzpicture}
\end{center}
These functions are defined by arrows that direct each element of $\Omega$ to the value of the function on this element. Formally, $u((n, a)) = (n + 1, a)$ and $v((n, a)) = (n + 1, 1 - a)$ for every $(n, a)\in\Omega$. It holds that $u, v\in\mathbb{M}$ because they both always increase the first coordinate by 1.

We also consider the following two functions $f_0, f_1\colon\Omega\to\Omega$: 
\begin{equation}
    \label{fg}
    f_b((n, a)) = \begin{cases} (n, a) & (n, a) = (0, 0), (0, 1) \mbox{ or } (1, b),\\(n + 1, a) & \mbox{otherwise},
    \end{cases} \qquad b\in\{0, 1\}
\end{equation}

For the reader's convenience, we depict them as well.
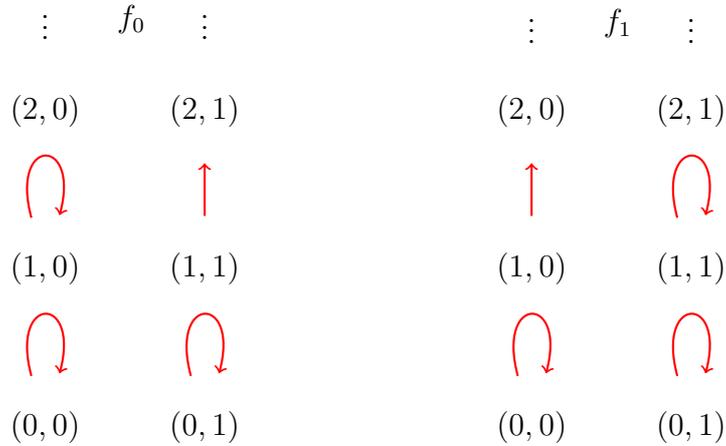
\begin{center}
\begin{tikzpicture}
\node[draw=none,  circle] (u00) {$(0, 0)$};
\node[draw=none,  circle, right = 0.7cm of u00] (u01) {$(0, 1)$};
\node[draw=none,  circle, above=0.7cm of u00] (u10) {$(1, 0)$};
\node[draw=none,  circle, right = 0.7cm of u10] (u11) {$(1, 1)$};
\node[draw=none,  circle, above=0.7cm of u10] (u20) {$(2, 0)$};
\node[draw=none,  circle, right = 0.7cm of u20] (u21) {$(2, 1)$};
\node[draw=none, above = 0.1cm of u20] (udots1) {$\vdots$};
\node[draw=none, above = 0.1cm of u21] (udots2) {$\vdots$};
\node[draw=none, right = 0.6cm of udots1] (u) {$f_0$};

\path
    (u00) edge [loop above, red,thick]  (u00);
\path
    (u01) edge [loop above, red,thick]  (u01);
\path
    (u10) edge [loop above, red,thick]  (u10);
\draw[->,red,thick] (u11) -- (u21);

\node[draw=none,  circle, right = 5cm of u00] (v00) {$(0, 0)$};
\node[draw=none,  circle, right = 0.7cm of v00] (v01) {$(0, 1)$};
\node[draw=none,  circle, above=0.7cm of v00] (v10) {$(1, 0)$};
\node[draw=none,  circle, right = 0.7cm of v10] (v11) {$(1, 1)$};
\node[draw=none,  circle, above=0.7cm of v10] (v20) {$(2, 0)$};
\node[draw=none,  circle, right = 0.7cm of v20] (v21) {$(2, 1)$};
\node[draw=none, above = 0.05cm of v20] (vdots1) {$\vdots$};
\node[draw=none, above = 0.05cm of v21] (vdots2) {$\vdots$};
\node[draw=none, right = 0.6cm of vdots1] (v) {$f_1$};

\path
    (v00) edge [loop above, red,thick]  (v00);
\path
    (v01) edge [loop above, red,thick]  (v01);
\path
    (v11) edge [loop above, red,thick]  (v11);
\draw[->,red,thick] (v10) -- (v20);
\end{tikzpicture}
\end{center}
To see that  $f_0, f_1\in\mathbb{M}$, observe that both functions have 3 fixed points, and at remaining  points, they act by increasing the first coordinate by 1. There could be a problem with monotonicity if below some fixed point there were a point which is not a fixed point. However, the sets of fixed points of $f_0$ and $f_1$ are downwards-closed w.r.t.~$\preceq$.


Consider the following arena.
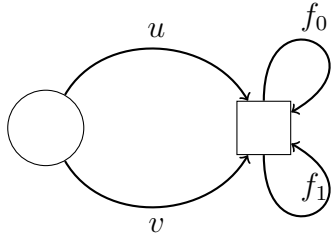
\begin{center}
\begin{tikzpicture}

\node[draw, circle,minimum size=1cm] (1) {};

\node[draw,  regular polygon, regular polygon sides=4, minimum size=1cm,right=2cm of 1] (2) {};

\path[->]
 (1) edge [thick, in=120,out=60,out distance=1cm,in distance=1cm] node[midway, above] {$u$} (2);

\path[->]
 (1) edge [thick, in=-120,out=-60,out distance=1cm,in distance=1cm] node[midway, below] {$v$} (2);

\path[->]
 (2) edge [thick, in=30,out=90,out distance=1.5cm,in distance=1.5cm] node[midway, above] {$f_0$} (2);

\path[->]
 (2) edge [thick, in=-30,out=-90,out distance=1.5cm,in distance=1.5cm] node[midway, above] {$f_1$} (2);
\end{tikzpicture}
\end{center}
The circle is controlled by Antagonist and the square is controlled by Protagonist. Assume that the game starts in the circle. We first show that Protagonist has a winning strategy w.r.t.~$\RL$. Then we show that Protagonist does not have a positional strategy which is winning w.r.t.~$\RL$. This implies that $\chrmem(\RL) > 1$.

Let us start with the first claim. After the first move of Antagonist, the ant moves either to $u(\zero) = (1, 0)$ or to  $v(\zero) = (1,1)$. In the first case, Protagonist wins by forever using the $f_0$-edge (the ant will always stay at $(1, 0)$). In the second case, Protagonist wins by always using the $f_1$-edge (the ant will always stay at $(1, 1)$).

Now we show that every positional strategy of Protagonist is not winning w.r.t.~$\RL$. In fact, there are just 2 Protagonist's positional strategies -- one which always uses the $f_0$-edge and the other which always uses the $f_1$-edge. The first one loses if Antagonist goes by the $v$-edge. Then the ant moves to $v(\zero) = (1, 1)$. If we start applying $f_0$ to the ant's position, the first coordinate of the ant will get arbitrarily large. Similarly, the second Protagonist's positional strategy loses if Antagonist goes by the $u$-edge.\qed
\end{proof}

\section{Upper Bound on the General Memory}
\label{sec:upp}
In this section, we establish 
\begin{proposition}
$\genmem(\RL) = 2$.
\end{proposition}
By Fact \ref{not_positional}, we only have to show an upper bound $\genmem(\RL) \le 2$. For that, for every arena $\mathcal{A}$ over $\mathbb{M}$ and for every Protagonist's strategy $S_1$ in $\mathcal{A}$ we construct a Protagonist's strategy $S_2$ with 2 states of general memory for which the following holds: 
for every node $v$ of $\mathcal{A}$, if $S_1$ is winning w.r.t.~$\RL$ from $v$, then so is $S_2$.

We will use the following notation. Take any finite path $p = e_1 \ldots e_m$ in $\mathcal{A}$. Define $\ant(p) = \col(e_m) \circ \ldots\circ\col(e_2)\circ\col(e_1)(\zero)$.
In other words, $\ant(p)$ is the position of the ant after the path $p$. In case when $p$ is a $0$-length path, we set $\ant(p) = \zero$. We also write $\layer(p)$ for the first coordinate of $\ant(p)$.

Let $U$ be the set of nodes of $\mathcal{A}$ from which $S_1$ is winning w.r.t.~$\RL$. By definition of $\RL$, for every $P\in\ip(S_1, U)$ there exists $N\in\mathbb{N}$ such that $\layer(p) \le N$ for every finite prefix $p$ of $P$.
The first step of our argument is to change the quantifiers here. That is, we obtain a strategy $S_1^\prime$ for which there exists some $N\in\mathbb{N}$ such that $\layer(p) \le N$ for every $p\in \fp(S_1^\prime, U)$.

We use an argument, similar to one which was used in~\cite{chatterjee2010generalized} to show finite-memory determinacy of multi-dimensional energy games. We call a play $p\in\fp(S_1, U)$ \emph{regular} if there exist two prefixes $q_1$ and $q_2$ of $p$ such that, first, $q_1$ is shorter than $q_2$, second, $\target(q_1) = \target(q_2)$, and third, $\ant(q_1) = \ant(q_2)$. In other words, $q_1$ and $q_2$ must lead to the same node in $\mathcal{A}$ and to the same position of the ant in $\Omega$. We stress that $q_2$ might coincide with $p$, but $q_1$ must be a proper prefix of $p$. If $p\in \fp(S_1, U)$ is not regular, then we call it \emph{irregular}.

First, we show that there are only finitely many irregular plays in the set $\fp(S_1, U)$. Note that any prefix of an irregular play is irregular. Thus, irregular plays form a collection of trees with finite branching (for each $u\in U$ there is a tree of irregular plays that start at $u$). Assume for contradiction that there are infinitely many irregular plays. Then, by K\H{o}nig's lemma, there exists an infinite branch in one of our trees. It gives some $P\in\ip(S_1, U)$ whose finite prefixes are all irregular. However, $P$ must be winning for Protagonist w.r.t.~$\RL$. In other words, there exists $N\in\mathbb{N}$ such that $\layer(p) \le N$ for every finite prefix $p$ of $P$. So, if $p$ ranges over finite prefixes of $P$, then $\ant(p)$ takes only finitely many values. Hence, there exist a node $v$ of $\mathcal{A}$ and some $(n, b)\in \Omega$ such that $v = \target(p)$ and $(n, b) = \ant(p)$ for infinitely many prefixes $p$ of $P$. Consider any two such prefixes. A longer one is regular because the shorter one is its prefix and leads to the same node in $\mathcal{A}$ and to the same position of the ant. This is a contradiction.

We now define $S_1^\prime$. It will maintain the following invariant for plays that start at $U$: if $p_{cur}$ is the current play, then there exists an irregular $p\in\fp(S_1, U)$ such that $\target(p_{cur}) = \target(p)$ and $\ant(p_{cur}) = \ant(p)$. Since there are only finitely many irregular plays, this invariant implies that $\ant(p_{cur})$ takes only finitely many values over $p_{cur}\in \fp(S_1^\prime, U)$, as required from $S_1^\prime$.

To maintain the invariant, $S_1^\prime$ plays as follows.
In the beginning, $p_{cur} = \lambda_w$ for some $w\in U$. Hence, we can set $p = \lambda_w$ also. Indeed, $\lambda_w\in \fp(S_1, U)$ and it is irregular as it has no proper prefixes. Let us now show how to maintain the invariant. Consider any play $p_{cur}$ with $S_1^\prime$ for which there exists an irregular $p\in\fp(S_1, U)$ such that $\target(p_{cur}) = \target(p)$ and $\ant(p_{cur}) = \ant(p)$. In this position, if its Protagonist's turn to move, $S_1^\prime$ makes the same move as $S_1$ from $p$. As a result, some edge $e$ is played. Observe that $pe\in\fp(S_1, U)$. In turn, our new current play with $S_1^\prime$ is $p_{cur} e$. We have that $\target(p_{cur}e) = \target(pe) = \target(e)$ and $\ant(p_{cur}e) = \col(e)\big(\ant(p_{cur})\big) = \col(e)\big(\ant(p)\big) = \ant(pe)$. So, 
if $pe$ is irregular, then the invariant is maintained. Now, assume that $pe$ is regular. Then there are two prefixes $q_1$ and $q_2$ of $pe$ such that, first, $q_1$ is shorter than $q_2$, second, $\target(q_1) = \target(q_2)$, and third, $\ant(q_1) = \ant(q_2)$. Since $p$ is irregular, $q_2$ cannot be a prefix of $p$. Hence, $q_2 = pe$. By the same reason, $q_1$ is irregular. Thus, invariant is maintained if we set  the new value of $p$ be $q_1$. Indeed, $\target(p_{cur}e) = \target(pe) = \target(q_2) = \target(q_1)$ and $\ant(p_{cur}e) = \ant(pe) =\ant(q_2) = \ant(q_1)$.

We now turn $S_1^\prime$ into a strategy $S_2$ with 2 states of general memory which is winning w.r.t.~$\RL$ from every node of $U$.

\textbf{Preliminary definitions.} Let $X$ be the set of nodes reachable from $U$ by plays with $S_1^\prime$.
 Next, for $v\in X$, define $\Omega_v\subseteq\Omega$ as the set of all  $(n, b)\in\Omega$ such that $(n, b) = \ant(p)$ for some $p\in\fp(S_1^\prime, W)$ with $v = \target(p)$. In other words, $\Omega_v$ is the set of all possible positions of the ant that can arise at $v$ if we play according to  $S_1^\prime$ from a node of $U$.

 Now, take any $v\in X$. The set $\Omega_v$ is non-empty and, by our requirements on $S_1^\prime$, finite. Hence, it has $1$ or $2$ maximal elements w.r.t.~$\preceq$. We will denote them by $M_0^v$ and $M_1^v$. If $\Omega_v$ has just a single maximal element, then $M_0^v = M_1^v$. If $\Omega_v$ has two different maxima, then let $M_0^v$ be the one having $0$ as the second coordinate. Finally, for every $v\in X$ and for every $b\in\{0, 1\}$ fix some $p_b^v\in\fp(S_1^\prime, U)$ such that $\target(p_b^v) = v$ and $\ant(p_b^v) = M_b^v$.

\textbf{Description of $S_2$.}
Two states of $S_2$ will  be denoted by $0$ and $1$. The initial state of $S_2$ is $0$.  The next-move function of $S_2$ is defined as follows. Assume that the state of $S_2$ is $I\in\{0,1\}$ and it has to make a move from a node $v$. If $v\notin X$, it makes an arbitrary move (this case does not matter for the argument below). Now, assume that $v\in X$. Then $S_2$ make the same move as $S_1^\prime$ after $p_I^v$.

We now describe the memory structure of $S_2$. Assume that it receives an edge $e$ when its state is $I\in\{0,1\}$. The new state $J\in\{0,1\}$ is computed as follows. Denote $u = \source(e)$ and $v = \target(e)$. If $u\notin X$ or $v\notin X$, then $J = 0$ (again, this case is irrelevant for the rest of the argument). Assume now that $u, v\in X$. If $\col(e)\big(M_I^u\big) \in \Omega_v$, then we find some  $b\in\{0, 1\}$ such that $\col(e)\big(M_{I}^u\big) \preceq M_b^v$ and set $J = b$.
Otherwise, we set $J = 0$.

\textbf{Showing that $S_2$ is winning from $W$.} First, we observe that $\target(p) \in X$ for every  $p\in\fp(S_2, U)$ (in other words, $S_2$ cannot leave $X$ if we start somewhere in $U$). Indeed, assume for contradiction that some play with $S_2$ leaves $X$ from some node $v\in X$. Let $I$ be the state of $S_2$ at the moment just before leaving $X$. If it is a Protagonist's turn to move, then it moves as $S_1^\prime$ after $p_I^v$. Recall that $p_I^v\in \fp(S_1^\prime, U)$. Thus, we obtain a continuation of $p_I^v$ which is consistent with $S_1^\prime$ and leads outside $X$. This contradicts the definition of $X$. Now, if it is an Antagonist's turn to move from $v$, then any continuation of $p_I^v$ by one edge is consistent with $S_1^\prime$, so we obtain the same contradiction.

Next, we show that for any play $p\in\fp(S_2, U)$ we have $\ant(p) \preceq M_I^{\target(p)}$, where $I$ is the state of $S_2$ after $p$. This statement implies that $S_2$ is winning w.r.t.~$\RL$ from every node of $U$. Note that $M_I^{\target(p)}$ is well-defined thanks to the previous paragraph.

We prove this statement by induction on the length of $p$. Let us start with the induction base. Assume that $|p| = 0$ (then $p = \lambda_w$ for some $w\in U$). The state of $S_2$ after $p$ is the initial state, that is, $0$. Thus, we have to show that $\ant(p) \preceq M_0^{\target(p)}$. Note that $p$ has length $0$ and hence is consistent with any strategy. In particular, $p\in \fp(S_1^\prime, U)$. Hence, $\ant(p) \in\Omega_{\target(p)}$. If $\Omega_{\target(p)}$ has just a single maximum, then $\ant(p)$ does not exceed this maximum, as required. Now, if $M_0^{\target(p)} \neq M_1^{\target(p)}$, then the second coordinate of $M_0^{\target(p)}$ is $0$, so we have  $\ant(p) \preceq M_0^{\target(p)}$ just because $\ant(p) = \zero$.

Next, we establish the induction step. Consider any $p\in\fp(S_2, U)$ of positive length and assume that for all paths from $\fp(S_2, U)$ of smaller length the statement is already proved. We prove our statement for $p$. Let $e$ be the last edge of $p$. Correspondingly, let $q$ be the part of $p$ preceding $e$. Denote $u = \target(q) = \source(e)$ and $v = \target(p) = \target(e)$

Any prefix of $p$ is also in $\fp(S_2, U)$, so $q\in \fp(S_2, U)$. Therefore, our statement holds for $q$. Namely, if $I$ is the state of $S_2$ after $q$, then  $\ant(q)  \preceq M_I^u$.

Let $J$ be the state of $S_2$ after $p$. Our goal is to show that $\ant(p) \preceq M_J^v$. Note that $\ant(p) = \col(e)\big(\ant(q)\big)$ by definition of $\ant$. Since $\col(e) \in\mathbb{M}$ is monotone and $\ant(q) \preceq M_I^u$, we have that $\ant(p) = \col(e)\big(\ant(q)\big) \preceq \col(e)\big(M_I^u\big)$. It remains to show that $\col(e)\big(M_I^u\big) \preceq M_J^v$. Note that $J$ is the state into which $S_2$ transits from the state $I$ after receiving $e$. By definition of the memory structure of $S_2$, it is sufficient to show that $\col(e)\big(M_I^u\big) \in \Omega_v$.

 By definition of $p_I^u$, we have that $M_I^u = \ant(p_I^u)$. Hence, $\col(e)\big(M_I^u\big) = \ant(p_I^u e)$. The path $p_I^u e$ starts at some node of $U$ and ends in $\target(e) = v$. Thus, to establish $\ant(p_I^u e)\in \Omega_v$, it remains to show consistency of $p_I^ue$ with $S_1^\prime$. We have $p_I^u\in \fp(S_1^\prime, U)$ by definition of $p_I^u$. In turn, if Protagonist is the one to move from $u = \target(p_I^u)$, then $e = S_1^\prime(p_I^u)$. Indeed, $e$ is the edge played by $S_2$ from $u$ when its state is $I$. Hence, $e = S_1^\prime(p_I^u)$, by the definition of the next-move function of $S_2$.

%
%
%

\begin{thebibliography}{8}
\bibitem{bloem2018graph}
{\sc Bloem, R., Chatterjee, K., and Jobstmann, B.}
\newblock {Graph games and reactive synthesis}.
\newblock In {\em Handbook of Model Checking}. Springer, 2018, pp.~921--962.


\bibitem{bouyer_et_al:LIPIcs.STACS.2022.16}
{\sc Bouyer, P., Randour, M., and Vandenhove, P.}
\newblock {Characterizing Omega-Regularity Through Finite-Memory Determinacy of
  Games on Infinite Graphs}.
\newblock {\em TheoretiCS}, 2 (2023), 1--48.

\bibitem{bouyer_et_al:LIPIcs:2020:12836}
{\sc Bouyer, P., Le~Roux, S., Oualhadj, Y., Randour, M., and Vandenhove, P.}
\newblock {Games Where You Can Play Optimally with Arena-Independent Finite
  Memory}.
\newblock {\em Logical Methods in Computer Science}, 18 (2022), 11:1-11:44.

\bibitem{buchi1969solving}
{\sc B{\"u}chi, J.~R., and Landweber, L.~H.}
\newblock {Solving sequential conditions by finite-state strategies}.
\newblock {\em Transactions of the American Mathematical Society 138\/} (1969),
  295--311.

\bibitem{regular}
{\sc Bouyer, P., Fijalkow, N., Randour, M., and Vandenhove, P.}
\newblock {How to Play Optimally for Regular Objectives?}
\newblock In {\em 50th International Colloquium on Automata, Languages, and Programming (ICALP 2023)\/}, vol.~261
  of {\em Leibniz International Proceedings in Informatics (LIPIcs)}, Schloss
  Dagstuhl -- Leibniz-Zentrum f{\"u}r Informatik, pp.~118:1--118:18.

\bibitem{casares:LIPIcs.CSL.2022.12}
{\sc Casares, A.}
\newblock {On the Minimisation of Transition-Based Rabin Automata and the
  Chromatic Memory Requirements of Muller Conditions}.
\newblock In {\em 30th EACSL Annual Conference on Computer Science Logic (CSL
  2022)\/} (Dagstuhl, Germany, 2022), F.~Manea and A.~Simpson, Eds., vol.~216
  of {\em Leibniz International Proceedings in Informatics (LIPIcs)}, Schloss
  Dagstuhl -- Leibniz-Zentrum f{\"u}r Informatik, pp.~12:1--12:17.

\bibitem{casares2022size}
{\sc Casares, A., Colcombet, T., and Lehtinen, K.}
\newblock {On the size of good-for-games Rabin automata and its link with the
  memory in Muller games}.
\newblock In {\em 349th International Colloquium on Automata, Languages, and Programming (ICALP 2022)\/}, vol.~229
  of {\em Leibniz International Proceedings in Informatics (LIPIcs)}, Schloss
  Dagstuhl -- Leibniz-Zentrum f{\"u}r Informatik, pp.~117:1--117:20.

\bibitem{chatterjee2010generalized}
{\sc Chatterjee, K., Doyen, L., Henzinger, T.~A., and Raskin, J.-F.}
\newblock {Generalized mean-payoff and energy games}.
\newblock In {\em IARCS Annual Conference on Foundations of Software Technology
  and Theoretical Computer Science (FSTTCS 2010)\/} (2010), Schloss
  Dagstuhl-Leibniz-Zentrum fuer Informatik.

\bibitem{colcombet2014playing}
{\sc Colcombet, T., Fijalkow, N., and Horn, F.}
\newblock {Playing safe}.
\newblock In {\em 34th International Conference on Foundation of Software
  Technology and Theoretical Computer Science (FSTTCS 2014)\/} (2014), Schloss
  Dagstuhl-Leibniz-Zentrum fuer Informatik.

\bibitem{dziembowski1997much}
{\sc Dziembowski, S., Jurdzinski, M., and Walukiewicz, I.}
\newblock {How much memory is needed to win infinite games?}
\newblock In {\em Proceedings of Twelfth Annual IEEE Symposium on Logic in
  Computer Science\/} (1997), IEEE, pp.~99--110.

\bibitem{kopczynski2007omega}
{\sc Kopczy{\'n}ski, E.}
\newblock {Omega-regular half-positional winning conditions}.
\newblock In {\em International Workshop on Computer Science Logic\/} (2007),
  Springer, pp.~41--53.

\bibitem{phdthesis}
{\sc Kopczy\'{n}ski, E.}
\newblock {\em Half-positional determinacy of infinite games}.
\newblock PhD thesis, Warsaw University, 2008.

\bibitem{kozachinskiy2022state}
{\sc Kozachinskiy, A.}
\newblock {State complexity of chromatic memory in infinite-duration games}.
\newblock {\em arXiv preprint arXiv:2201.09297\/} (2022).

\bibitem{le2020time}
{\sc Le~Roux, S.}
\newblock {Time-aware uniformization of winning strategies}.
\newblock In {\em Conference on Computability in Europe\/} (2020), Springer,
  pp.~193--204.

\bibitem{zielonka1998infinite}
{\sc Zielonka, W.}
\newblock {Infinite games on finitely coloured graphs with applications to
  automata on infinite trees}.
\newblock {\em Theoretical Computer Science 200}, 1-2 (1998), 135--183.
\end{thebibliography}
%

\section{Lower Bound on the Chromatic Memory}
\label{sec:low}
In this section, we establish the following proposition.
\begin{proposition}
There exists a finite set $C\subseteq \mathbb{M}$ such that $\chrmem(\RL\cap C^\omega) = +\infty$.
\end{proposition}
We start by describing $C$. First, we put there $f_0, f_1$ that are defined in \eqref{fg}. Next, put there a function $h\colon\Omega\to\Omega$, defined by
\[h((n, a)) = \begin{cases}(n - 1, a) & n > 1 \\ (0, 0) & n = 0, 1.\end{cases}\]

For the reader's convenience, it is depicted below:
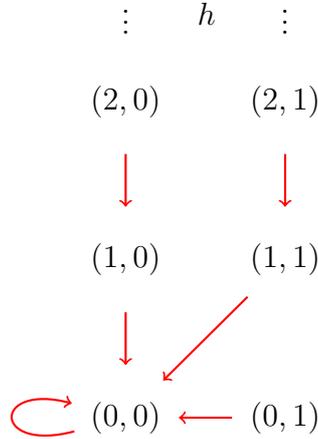
\begin{center}
\begin{tikzpicture}

\node[draw=none,  circle] (h00) {$(0, 0)$};
\node[draw=none,  circle, right = 0.7cm of h00] (h01) {$(0, 1)$};
\node[draw=none,  circle, above=0.7cm of h00] (h10) {$(1, 0)$};
\node[draw=none,  circle, right = 0.7cm of h10] (h11) {$(1, 1)$};
\node[draw=none,  circle, above=0.7cm of h10] (h20) {$(2, 0)$};
\node[draw=none,  circle, right = 0.7cm of h20] (h21) {$(2, 1)$};
\node[draw=none, above = 0.05cm of h20] (hdots1) {$\vdots$};
\node[draw=none, above = 0.05cm of h21] (hdots2) {$\vdots$};
\node[draw=none, right = 0.6cm of hdots1] (h) {$h$};

\path
    (h00) edge [loop left, red,thick]  (h00);
    
    \draw[->,red,thick] (h01) -- (h00);

\draw[->,red,thick] (h10) -- (h00);
\draw[->,red,thick] (h11) -- (h00);
\draw[->,red,thick] (h20) -- (h10);
\draw[->,red,thick] (h21) -- (h11);
\end{tikzpicture}
\end{center}

Let us establish that $h\in\mathbb{M}$. Take any $(n, a), (m, b) \in\Omega$ such that $(n, a)\preceq (m, b)$. We show that $h((n, a)) \preceq h((m, b))$. If $(n, a) = (m, b)$, then $h((n, a)) = h((m, b))$. Now, if $(n, a)\neq (m, b)$, then $n < m$. The first coordinates of $h((n, a))$ and $h((m, b))$ are $\max\{0, n - 1\}$ and $\max\{0, m - 1\}$, respectively. If $m > 1$, then $\max\{0, m - 1\} = m - 1 > \max\{0, n - 1\}$, which implies that $h((n, a)) \preceq h((m, b))$. Now, if $m \le 1$, then $n\le 1$ also, which means that $h((n, a)) = h((m, b)) = (0, 0)$.

We will also put into $C$ two more functions $p^0, p^1\colon\Omega\to\Omega$, but to define them, we need some auxiliary work. 
\begin{definition}
A function $f\colon\Omega\to\Omega$ is called \textbf{incremental} if for every $(n, b)\in\Omega$ there exists $c\in\{0,1\}$ such that  $f((n, b)) = (n + 1, c)$.
\end{definition}
Let us first observe that every incremental $f\colon\Omega\to\Omega$ belongs to $\mathbb{M}$. Indeed, take any $(n, a), (m, b)\in \Omega$ such that $(n, a)\preceq (m, b)$. We have to show that $f((n, a))\preceq f((m, b))$. for every incremental $f$. If $(n, a) = (m, b)$, we have $f((n, a)) = f((m, b))$ . Otherwise, $n < m$. Then $f((n, a)) = (n + 1, c)$ and $f((m, b)) = (m + 1, d)$ for some $c, d\in\{0,1\}$. Since $n +1 < m +1$, we have $f((n, a))\preceq f((m, b))$.

We say that two binary words $x, y\in\{0, 1\}^*$ are \emph{$Q$-indistinguishable}  if there exists no deterministic finite automaton over $\{0, 1\}$  with at most $Q$ states which comes to different states on $x$ and on $y$.

\begin{lemma} 
\label{sequences}
There exist two infinite sequences of bits  $\{I_n^0\}_{n = 0}^\infty\in\{0,1\}^\omega $ and $\{I_n^1\}_{n = 0}^\infty\in\{0,1\}^\omega$ such that for every $Q \in\mathbb{N}$ there exist $t\in\mathbb{N}$ and two $Q$-indistinguishable binary words $x = x_0\ldots x_{t - 1}$ and $y = y_0\ldots y_{t - 1}$ of length $t$ such that:
\[I_0^{x_0} \oplus \ldots \oplus I_{t-1}^{x_{t-1}} \neq I_0^{y_0} \oplus \ldots \oplus I_{t-1}^{y_{t-1}}\]
($\oplus$ denotes XOR). 
\end{lemma}
\begin{proof}
Let $A_Q$ be the number of deterministic finite automata over $\{0,1\}$ with at most $Q$ states. For $Q\in\mathbb{N}$, let $l_Q$ be any number such that $2^{l_Q} > Q^{A_Q}$.

We split natural numbers in consecutive blocks $B_1, B_2, B_3, \ldots$, where
\[B_Q = \{l_1 + \ldots + l_{Q - 1}, \ldots, l_1 + \ldots + l_{Q - 1} + l_Q - 1\}\]
(so that $|B_Q| = l_Q$). We first define $I^0_n, I^1_n$ for $n\in B_1$, then for $n\in B_2$, and so on. The requirement of the lemma for $Q$ will be guaranteed after we define our sequences in the first $Q$ blocks.

More specifically, assume that our sequences are already defined in the first $Q - 1$ blocks. We have to define them in $B_Q$ in some way that satisfies the requirement of the lemma for $Q$. Since $2^{l_Q} >  Q^{A_Q}$, there exist two different $Q$-indistinguishable binary words  $a, b\in\{0,1\}^{l_Q}$. Indeed, to every binary word $w$ we can assign a tuple, where for all deterministic finite automaton $\mathcal{A}$ with at most $Q$ states we have a coordinate, indicating the state of $\mathcal{A}$ after reading $w$. The number of such tuples is bounded by  $Q^{A_Q}$. Hence, in $\{0,1\}^{l_Q}$ there are two different binary words $a, b$ with the same tuple assigned to them. This means that $a$ and $b$ are $Q$-indistinguishable.

Let $t = l_1 +\ldots + l_Q$ and
\[x = \underbrace{00\ldots0}_{l_1+\ldots +l_{Q - 1}} a\in\{0,1\}^t, \qquad y = \underbrace{00\ldots0}_{l_1+\ldots +l_{Q - 1}} b\in\{0,1\}^t.\]
Note that $x$ and $y$ are obtained from $a$ and $b$ by attaching the same prefix. Hence, $x$ and $y$ are also $Q$-indistinguishable. We claim that we can define $I^0_n, I^1_n$ for $n\in B_Q$ in such a way that:
\begin{equation}
\label{ineq}
I_0^{x_0} \oplus \ldots \oplus I_{t-1}^{x_{t-1}} \neq I_0^{y_0} \oplus \ldots \oplus I_{t-1}^{y_{t-1}}.
\end{equation}
Since $a$ and $b$ are different, we have that $x$ and $y$ are also different. But both $x$ and $y$ start with 
$l_1+\ldots +l_{Q - 1}$ zeros. Hence, all the indices where they differ belong to $B_Q$. Take any $m\in B_Q$ such that $x_m\neq y_m$. Define $I^0_n = I^1_n = 0$ for all $n\in B_Q\setminus\{m\}$. It remains to define $I^0_m$ and $I^1_m$ in such a way that \eqref{ineq} holds. Note that all the summands except $I^{x_m}_m$ and $I^{y_m}_m$ are already defined. Since $x_m\neq y_m$, one of these summands is  $I^0_m$ and the other is $I^1_m$. One of them is in the left-hand side, and the other one is in the right-hand side. Hence, we can define them in such a way that the inequality is true.\qed
\end{proof}
We now take sequences $\{I_n^0\}_{n = 0}^\infty\in\{0,1\}^\omega $ and $\{I_n^1\}_{n = 0}^\infty\in\{0,1\}^\omega$, satisfying Lemma \ref{sequences}, and define $p^0, p^1\colon\Omega\to\Omega$ as follows:
\[
p^0((n, b)) = \Big(n + 1,  b\oplus I_n^0\Big),\qquad p^1((n, b)) = \Big(n + 1,  b\oplus I_n^1\Big).
\]
Note that $p^0, p^1$ are incremental. Hence, $p^0, p^1\in\mathbb{M}$. We set $C = \{f_0, f_1, h, p^0, q^1\}$ and show that $\chrmem(\RL\cap C^\omega) =+\infty$. For that, for every $Q\in\mathbb{N}$, we show that $\chrmem(\RL\cap C^\omega)  > Q$. Fix any $Q\in\mathbb{N}$ and let $t\in\mathbb{N}$ and $x = x_0\ldots x_{t - 1}\in\{0,1\}^t,y = y_0\ldots y_{t - 1}\in\{0, 1\}^t$ be such that $x$ and $y$ are $Q$-indistinguishable and 
\begin{equation}
\label{I}
I_0^{x_0} \oplus \ldots \oplus I_{t-1}^{x_{t-1}} \neq I_0^{y_0} \oplus \ldots \oplus I_{t-1}^{y_{t-1}}
\end{equation}
(existence of such $t$, $x$ and $y$ is guaranteed by Lemma \ref{sequences}). Consider the following arena:
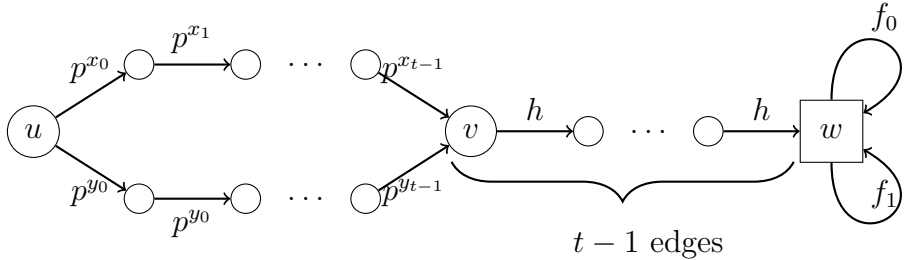
\begin{center}
\begin{tikzpicture}

\node[draw, circle,minimum size=0.5cm] (1) {$u$};

\node[draw, circle, above right=0.5cm and 1cm of 1] (2) {};

\node[draw, circle, below right=0.5cm and 1cm of 1] (3) {};

\node[draw, circle, right=1cm of 2] (4) {};
\node[draw, circle, right=1cm of 3] (5) {};

\node[draw=none, right=0.2cm of 4] (6) {$\ldots$};
\node[draw=none, right=0.2cm of 5] (7) {$\ldots$};

\node[draw, circle, right=0.2cm of 6] (8) {};
\node[draw, circle, right=0.2cm of 7] (9) {};

\node[draw, circle, below right=0.5cm and 1cm of 8,minimum size=0.5cm] (10) {$v$};

\node[draw, circle, right= 1cm of 10] (11) {};
\node[draw=none, right=0.2cm of 11] (12) {$\ldots$};
\node[draw, circle, right= 0.2cm of 12] (13) {};
\node[draw,  regular polygon, regular polygon sides=4, minimum size=0.5cm,right=1cm of 13] (14) {$w$}; 
\path[->]
 (1) edge [thick] node[midway, above] {$p^{x_0}$} (2);

\path[->]
 (1) edge [thick] node[midway, below] {$p^{y_0}$} (3);

\path[->]
 (2) edge [thick] node[midway, above] {$p^{x_1}$} (4);

\path[->]
 (3) edge [thick] node[midway, below] {$p^{y_0}$} (5);

\path[->]
 (8) edge [thick] node[midway, above] {$p^{x_{t-1}}$} (10);
\path[->]
 (9) edge [thick] node[midway, below] {$p^{y_{t-1}}$} (10);

\path[->]
 (10) edge [thick] node[midway, above] {$h$} (11);

\path[->]
 (13) edge [thick] node[midway, above] {$h$} (14);

\path[->]
 (14) edge [thick, in=30,out=90,out distance=1.5cm,in distance=1.5cm] node[midway, above] {$f_0$} (14);

\path[->]
 (14) edge [thick, in=-30,out=-90,out distance=1.5cm,in distance=1.5cm] node[midway, above] {$f_1$} (14);

\draw [decorate,
	decoration = {brace,mirror,amplitude=15pt}, thick] (5.5,-0.4) --  (10,-0.4);

\node[draw=none, below=1cm of 12] (16) {$t - 1$ edges};
\end{tikzpicture}
\end{center}
All circles are controlled by Antagonist and the square is controlled by Protagonist. The game starts at the node $u$. We claim that Protagonist has a winning strategy w.r.t.~$\RL$. Indeed, in the beginning, Antagonist has two choices -- to go through $p^{x_0}, p^{x_1},\ldots, p^{x_{t-1}}$ or to go through $p^{y_0}, p^{y_1},\ldots, p^{y_{t-1}}$. In any case, upon reaching $v$, the first coordinate of the ant will be $t$ (both $p^0$ and $p^1$ always increase the first coordinate of the ant by 1). Then we go through $t - 1$ edges colored by $h$. As a result, the position of the ant at $w$ will be either $(1, 0)$ or $(1, 1)$. If it is $(1, 0)$, then Protagonist wins by always using $f_0$. If it is $(1, 1)$, then Protagonist wins by always using $f_1$.

It remains to show that Protagonist has no winning strategy with at most $Q$ states of chromatic memory. Indeed, consider any Protagonist's strategy $S$ with at most $Q$ states of chromatic memory. Our goal is to show that $S$ is not winning. It is built on top of some chromatic memory structure with at most $Q$ states. This memory structure, by definition, only reads colors of edges. Hence, when we go from $u$ to $v$, we either feed $p^{x_0}p^{x_1}\ldots p^{x_{t-1}}\in\{p^0, p^1\}^t$ or $p^{y_0}p^{y_1}\ldots p^{y_{t-1}}\in\{p^0, p^1\}^t$ to it. We claim that $S$ comes into the same state on these two sequences. Indeed, up to renaming letters of the alphabet, we may assume that we feed $x = x^0\ldots x^{t-1}$ and $y = y^0\ldots y^{t-1}$ to the memory structure of $S$. By definition, $x$ and $y$ are $Q$-indistinguishable. Since the memory structure of $S$ has at most $Q$ states, it must come into the same state on $x$ and $y$. We conclude the state of $S$ at $v$, and hence at $w$, will be the same in both possible plays. Thus, $S$ acts identically at $w$ in these two plays.

At the same time, there are two different possible positions of the ant at $w$. More specifically, if the Antagonist goes through $p^{x_0}p^{x_1}\ldots p^{x_{t-1}}$, the position of the ant will be 
\begin{align*}
\underbrace{h\circ\ldots \circ h}_{t-1}\circ p^{x_{t-1}}\circ\ldots\circ p^{x_1}\circ p^{x_0}(\zero) &= \underbrace{h\circ\ldots \circ h}_{t-1}\circ p^{x_{t-1}}\circ\ldots\circ p^{x_1}((1, I^{x_0}_0)) \\
&= \underbrace{h\circ\ldots \circ h}_{t-1}((t, I_0^{x_0} \oplus \ldots \oplus I_{t-1}^{x_{t-1}})\\
&= (1, I_0^{x_0} \oplus \ldots \oplus I_{t-1}^{x_{t-1}}).
\end{align*}
Likewise, if the Antagonist goes through $p^{y_0}p^{y_1}\ldots p^{y_{t-1}}$, the position of the ant will be 
\[\underbrace{h\circ\ldots \circ h}_{t-1}\circ p^{x_{t-1}}\circ\ldots\circ p^{x_1}\circ p^{x_0}(\zero) = (1, I_0^{y_0} \oplus \ldots \oplus I_{t-1}^{y_{t-1}}).\]

Since $I_0^{x_0} \oplus \ldots \oplus I_{t-1}^{x_{t-1}}\neq I_0^{y_0} \oplus \ldots \oplus I_{t-1}^{y_{t-1}}$ by \eqref{I}, we conclude that both $(1, 0)$ and $(1, 1)$ are possible positions of the ant at $w$. But once again, $S$ acts in the same way at $w$ in both cases.
Assume first that $S$ plays the $f_0$-edge when it first reaches $w$. Then $S$ loses if the ant reached $w$ being in $(1, 1)$. Indeed, after $S$ plays its first move at $w$, the position of the ant becomes $f_0((1, 1)) = (2, 1)$. If the first coordinate of the ant is 2 or more, both $f_0$ and $f_1$ increase it by 1. Hence, no matter what Protagonist does afterwards, the ant will get infinitely high in $\Omega$. Likewise, if the first move of $S$ at $w$ is the $f_1$-edge, then it loses if the ant reaches $w$ being in $(1, 0)$.

\end{document}